\documentclass{article}

\usepackage{graphicx} 
\usepackage{amsmath,amssymb,amsthm}
\usepackage{booktabs}
\usepackage{xcolor}
\usepackage{enumitem}
\usepackage{mathtools}
\usepackage{multicol}
\usepackage[pdftex,
pdfstartview=FitH,
pdfpagetransition=Glitter,
pdfpagemode=None,
breaklinks=true,
colorlinks=true,
urlcolor=blue,
citecolor=magenta,
setpagesize=false]{hyperref}
\hypersetup{hidelinks}
\usepackage{subcaption}
\graphicspath{{Figures/}}
\usepackage{verbatim}
\usepackage{tikz}
\newtheorem{proposition}{Proposition}
\newtheorem{remark}{Remark}
\newtheorem{definition}{Definition}
\newtheorem{theorem}{Theorem}

\newcommand{\D}{\mathrm{d}}

\begin{document}
	
	\title{Constructing Trinomial Models Based\\ on Cubature Method on Wiener Space: Applications to Pricing Financial Derivatives}
	
	
	\author{
		Hossein Nohrouzian
		\and
		Anatoliy Malyarenko
		\and
		Ying Ni
	}
	
	\date{}
	\maketitle
	
	

\begin{abstract}
		This contribution deals with an extension to our developed novel cubature methods of degrees 5 on Wiener space. In our previous studies, we have shown that the cubature formula is exact for all multiple Stratonovich integrals up to dimension equal to the degree. In fact, cubature method reduces solving a stochastic differential equation to solving a finite set of ordinary differential equations. Now, we apply the above methods to construct trinomial models and to price different financial derivatives. We will compare our numerical solutions with the Black’s and Black--Scholes models’ analytical solutions. The constructed model has practical usage in pricing American-style derivatives and can be extended to more sophisticated stochastic market models.\\
		\textbf{keywords:} Cubature method, Stratonovich integral, Wiener space, stochastic market model
\end{abstract}

\section{Introduction and outline of the paper}
\label{Introduction and outline of the paper}
In mathematical finance, it  is common to describe the random changes in risky asset prices by  stochastic differential equations (SDEs). SDEs can be re-written in their integral forms. However, it is not possible to calculate all stochastic integrals in closed form. Therefore,  proper numerical methods should be used to estimate the value of such stochastic integrals.

One of the most popular numerical method to estimate stochastic integrals is Monte Carlo method (estimate). In particular, according to~\cite{paperE}, cubature methods and consequently cubature formulae construct a probability measure with finite support on a finite-dimensional real linear space which approximates the standard Gaussian measure. A generalisation of this idea, when a finite-dimensional space is replaced with the Wiener space, can be used for constructing modern Monte Carlo estimates  (see~\cite{BayerTeichmann2008} for the exact sense of modern Monte Carlo estimate). The idea of cubature method on Wiener space, among others, was developed in~\cite{LyonsVictoir2004}. The extension of this idea were developed and studied in~\cite{paperA,paperC,paperD,paperE}.

Our objective is to use cubature method and a cubature formula of degree 5 on Wiener space to estimate the expected values of functionals defined on the solutions of SDEs. This means that we use an extension to our developed novel cubature methods of degrees 5 to estimate the (discounted) expected values of European call and put payoff functions defined on the solutions of Black--Scholes and Black's SDEs. This extension includes a construction of a recombining trinomial tree model. The underlying asset prices in Black--Scholes and Black's models are log-normally distributed and the price dynamics follow  geometric Brownian motion. Moreover, both models have closed form solutions to find the price of European call and put options. Availability of closed form solutions of these models provides us with an opportunity to investigate if the sequence of our  trinomial models converges to a geometric Brownian motion or not.
Also, we can compare our numerical results with analytical ones and consequently estimate the corresponding errors of our method.
We would like to emphasize that the constructed trinomial tree has practical usage and applications in pricing path-dependent and American-style options.

The outline of this paper is as follows. In Section~\ref{Cubature formula in Black--Scholes and Black's model}, we briefly look at the cubature method on Wiener space and at the applications of cubature formula in Black--Scholes and Black's models. Then, in Section~\ref{Constructing a trinomial model via cubature formula}, we construct a trinomial model based on cubature formula on Black--Scholes model. After that,  in Section~\ref{Convergence to geometric Brownian motion}, we study the convergence of the sequences of constructed trinomial model to a geometric Brownian motion. In Section~\ref{Martingale probability measure}, we will study the  conditions which makes the probability measure in our trinomial model a martingale measure, i.e., risk-neutral probability measure. Later, in Section~\ref{Extension of the results and examples}, we will extend our results for more  cases and give some concrete examples where we study the behaviour of  the constructed trinomial model in Black--Scholes and Black's model. Finally, we close this paper by a discussion section.

\section{Cubature formula in Black--Scholes and Black's model}
\label{Cubature formula in Black--Scholes and Black's model}
In this section, we briefly review how the cubature method on Wiener space can be used in the Black--Scholes and Black's models.
\subsection{Black--Scholes model via cubature formula}
\label{Black--Scholes model via cubature formula}
Given a filtered probability space $(\Omega,\mathfrak{F},\mathsf{P},(\mathfrak{F})_{t\geq 0})$, let $S(t)$ be  time-$t$   price of a (non-dividend-paying) risky asset, $r$ and $\sigma$ be drift (risk-free interest rate) and diffusion (volatility of asset price) coefficients and  $\{W(t)\}_{t\geq 0}$ be the standard one-dimensional Wiener process.
The dynamics of  risky asset prices  in the Black--Scholes model~\cite{BlackScholes1973,Merton1973}, under equivalent martingale probability measure  $\mathsf{Q}\sim\mathsf{P}$,  satisfies the following SDE (originally proposed by Samuelson~\cite{Samuelson1965})
\begin{equation}
	\label{Samuelson SDE}
	\D S(t)= rS(t)\D t + \sigma S(t)\D W(t), \qquad 0\leq t\leq T.
\end{equation}

The solution to the SDE~\eqref{Samuelson SDE} is an It\^{o} process. After applying the Stratonovich correction, the above equation can be written in its Stratonovich form (see~\cite{Oksendal2013})

\begin{equation}
	\label{stratonovich SDE}
	\D S(t) = (r-\frac{1}{2}\sigma^2)S(t)\D t+ \sigma S(t)\circ\D W(t).
\end{equation}

We note that, cubature formula is valid in time interval with length one. Applying the results of cubature method on Wiener space in Equation~\eqref{stratonovich SDE} yields to~\cite{paperC,paperE}
\begin{align*}
	\D S_k(t) & = (r-\frac{1}{2}\sigma^2) S_k(t)\D t+ \sigma S_k(t) \D\omega_k(t)t, \qquad 0\leq t\leq 1,
\end{align*}
where  $\omega_k,$ \; $(1\leq k\leq l)$ is the $k$th  possible trajectory, and $l,\; (l\in \mathbb{Z}^+)$ stands for the number of trajectories in the cubature formula of degree $M$.

Rearranging the  last equation and calculating the integral of both hand sides gives
\begin{align}
	\label{Generalized stratonovich}
	\hat{S}_k(t_{j})  = \hat{S}_k(t_{j-1})\exp\left\{(r-\frac{1}{2}\sigma^2) [t_{j}-t_{j-1}]+ \sigma  [\omega_k(t_{j})-\omega_k(t_{j-1})]\right\},
\end{align}
with  $j=1,\ldots, l$ and $0\leq t_j\leq 1$.
In a cubature formula of degree $M=5$, the number of trajectories is $l=3$ and one of the possible solutions becomes
\begin{equation}
	\label{generalized cubature of O5}
	\omega_k(t_{j}) = 3\theta_{k,j}(t_{j}-t_{j-1})+\omega_k(t_{j-1}), \quad  j=1,2,3, \quad\omega_k(0) = 0,
\end{equation}
where $0=t_0<t_1<t_2<t_3=1$, i.e., the trajectories start from time $0$ and stop at time $1$,  $j=1, 2, 3$,   $t_{j}-t_{j-1}=1/3$ and with weight $\lambda_k$ and coefficients $\theta_k$ summarized in Table~\ref{Solutions O5}.
\begin{table}[t!]
	\centering
		\begin{tabular}{ccccc}
			\toprule
			$k$ \quad&\quad $\lambda_k$ \quad&\quad $\theta_{k,1}=\theta_{k,3}$ \quad&\quad $\theta_{k,2}$ \quad&\quad $\theta_{k,3}=\theta_{k,1}$ \\
			\midrule
			1 & $1/6$ &\quad $(-2\sqrt{3}\mp\sqrt{6})/{6}$  \quad&\quad $(-\sqrt{3}\pm\sqrt{6})/{3}$ \quad&\quad  $(-2\sqrt{3}\mp\sqrt{6})/{6}$ \\
			2 & $2/3$ & $\pm\sqrt{6}/{6}$ & $\mp\sqrt{6}/{3}$	&  $\pm\sqrt{6}/{6}$\\
			3 & $1/6$ & $(2\sqrt{3}\pm\sqrt{6})/{6}$  	& $(\sqrt{3}\mp\sqrt{6})/{3}$ 	&  $(2\sqrt{3}\pm\sqrt{6})/{6}$  \\
			\bottomrule
		\end{tabular}
		\caption{Information for cubature formulae of degree 5.}
		\label{Solutions O5}
\end{table}
	
Using Equation~\eqref{generalized cubature of O5}, we calculate  $\omega_3(t_3) =-\omega_1(t_3)=\sqrt{3}$ and $\omega_2(t_3) = 0$. For simplicity, denote $\omega_k(t_3)$ by $\omega_k$ Now, let us ignore the intermediate partitions in $0=t_0<t_1<t_2<t_3=1$. This reduces Equation~\eqref{Generalized stratonovich} to
\begin{align}
		\label{simplified stratonovich}
		\hat{S}_k(1)  = \hat{S}_{k}(0)\exp\left\{(r-\frac{1}{2}\sigma^2) + \sigma  \omega_k\right\}.
\end{align}

Equation~\eqref{simplified stratonovich} works for the time interval of length one, i.e. $t\in[0,1]$.
Let $T$ be the time to maturity for an option. Right now, we can only consider  $T=1$. We would like  considering yearly  interest rate, yearly volatility and yearly time to maturity in a more flexible way. So, we modify the drift and diffusion terms in Equation \eqref{simplified stratonovich}. This modification yields to
\begin{align}
		\label{modified stratonovich}
		\hat{S}_k(T)  = \hat{S}_{k}(0)\exp\left\{(r-\frac{1}{2}\sigma^2)T + \sigma\sqrt{T}  \omega_k\right\}.
\end{align}

Now, it is easy to calculate $\hat{S}_k(T)$ for $k = 1,2,3$. Then, the price of an option will be equal to its discounted expected payoff. For example, the price   of an European call option ($\pi_c$)  and the price   of an European put option ($\pi_p$) with strike price $K$ will be given by
\begin{align*}
\pi_c &= {e^{-rT}}\sum_{k=1}^3 \lambda_k\max\left(\hat{S}_k(T)-K,0\right),\\
\pi_p &= {e^{-rT}}\sum_{k=1}^3 \lambda_k\max\left(K-\hat{S}_k(T),0\right).
\end{align*}

Equation~\eqref{modified stratonovich}  works better for very small time to maturity. Later in this paper, we will try to improve the performance of our model by extending it into a trinomial tree model.
\subsection{Black's model via cubature formula}
In the Black's model~\cite{Black1976}, the dynamics of  risky forward rate prices  ${F(t)}$  satisfies
\begin{equation*}
	\D F(t)=  \sigma F(t)\D W(t),\quad 0\leq t\leq T.
\end{equation*}

The Stratonovich form of the above SDE becomes
\begin{equation*}
	\D F(t) =-\frac{1}{2}\sigma^2F(t)\D t+ \sigma F(t)\circ\D W(t).
\end{equation*}

Following the same procedure  as  in Section~\ref{Black--Scholes model via cubature formula}, we  get
\begin{align}
	\label{modified stratonovich for forward rate}
	\hat{F}_k(T)  = \hat{F}_{k}(0)\exp\left\{-\frac{1}{2}\sigma^2T+\sigma\sqrt{T}\omega_k \right\}.
\end{align}

\section{Constructing a trinomial tree via cubature formula}
\label{Constructing a trinomial model via cubature formula}
In this section, we develop, modify and revise the idea of constructing a trinomial model via cubature formula on Wiener space presented in~\cite{paperC}.
The idea of constructing an $n$-step trinomial tree is to divide the time interval $[0,T]$ by $n$ steps.
In each step, the price can
go up by an amount of $u_0$ with probability $p_u$,
go in middle by an amount of $m_0$ with probability $p_m$ and
go down by an amount $d_0$ with probability $p_d$.
Also, $0\leq p_u,p_m,p_d\leq 1$, $p_u+p_m+p_d=1$ and $u_0>m_0>d_0$.  In other words, we create  more trajectories (paths) for the underlying process in discrete time in order to get more accurate possible prices.
If a trinomial model is  recombining, then the number of nodes in the  constructed recombining trinomial tree is $2n+1$. Using  trinomial expansion, we can find all  $(n+1)(n+2)/2$ trinomial coefficients which represent the number of possible paths to reach each node.
Note that the sum of all trinomial coefficients is equal to the number of all possible paths in the constructed trinomial tree, that is, to $3^n$.
Moreover, we can simply calculate the corresponding probabilities and prices at each step of the tree and for each node. Figure~\ref{fig} illustrates the idea, where the dashed-arrows represent a recombining binomial tree (e.g., see CRR binomial model in~\cite{CRR1979}). We observe that for example in node $(3,4)$, $S_0u_0m_0^2=S_0u_0^2d_0$ and in node $(3,3)$ $S_0m_0^3=S_0u_0m_0d_0$.

The constructed tree can (among other possible applications) be used for pricing European, American-style  and path-dependent derivative options.
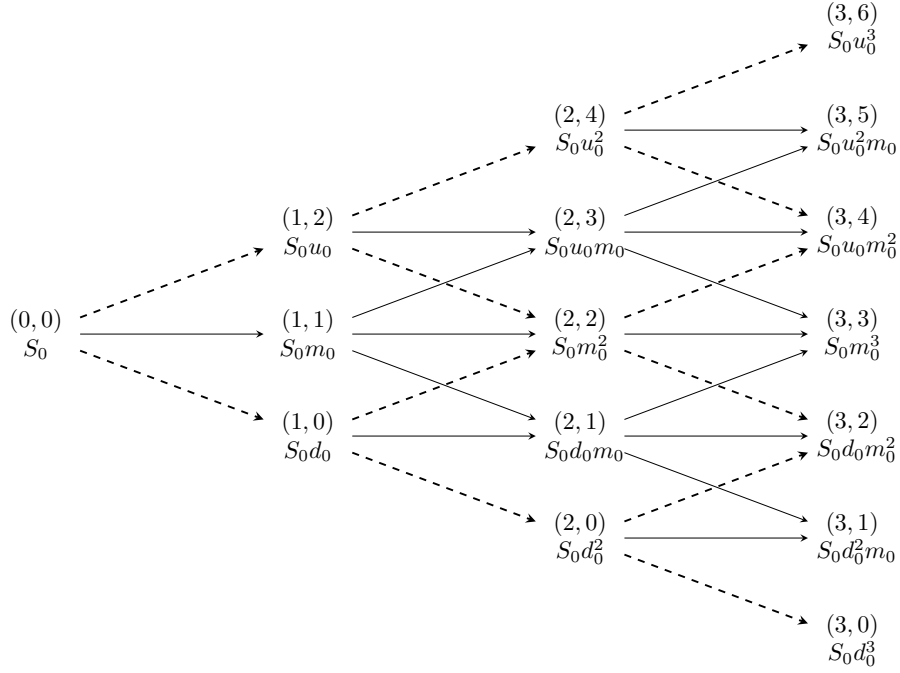
\begin{figure}[t!]
	\centering
	\tikzstyle{bag} = [text width=3em, text centered]
	\tikzstyle{end} = []
	\resizebox{\linewidth}{!}{%
	\begin{tikzpicture}[sloped, >=stealth]
		\node (a) at ( 0,0) [bag] {$(0,0)$\\$S_0$};
		\node (b) at ( 4,1.5) [bag] {$(1,2)$\\$S_0u_0$};
		\node (c) at ( 4,0) [bag] {$(1,1)$\\$S_0m_0$};
		\node (d) at ( 4,-1.5) [bag] {$(1,0)$\\$S_0d_0$};
		\node (e) at ( 8,3) [bag] {$(2,4)$\\$S_0u_0^2$};
		\node (f) at ( 8,1.5) [bag] {$(2,3)$\\$S_0u_0m_0$};
		\node (g) at ( 8,0) [bag] {$(2,2)$\\$S_0m_0^2$};
		\node (h) at ( 8,-1.5) [bag]{$(2,1)$\\$S_0d_0m_0$};
		\node (i) at ( 8,-3) [bag] {$(2,0)$\\$S_0d_0^2$};
		\node (j) at ( 12,4.5) [bag] {$(3,6)$\\$S_0u_0^3$};
		\node (k) at ( 12,3) [bag] {$(3,5)$\\$S_0u_0^2m_0$};
		\node (l) at ( 12,1.5) [bag] {$(3,4)$\\$S_0u_0m_0^2$};
		\node (m) at ( 12,0) [bag] {$(3,3)$\\$S_0m_0^3$};
		\node (n) at ( 12,-1.5) [bag] {$(3,2)$\\$S_0d_0m_0^2$};
		\node (o) at ( 12,-3) [bag]{$(3,1)$\\$S_0d_0^2m_0$};
		\node (p) at ( 12,-4.5) [bag] {$(3,0)$\\$S_0d_0^3$};
		\draw [dashed,->,thick] (a) to node [     ] {} (b);
		\draw [->] (a) to node [     ] {} (c);
		\draw [dashed,->,thick] (a) to node [     ] {} (d);
		\draw [dashed,->,thick] (b) to node [     ] {} (e);
		\draw [->] (b) to node [     ] {} (f);
		\draw [dashed,->,thick] (b) to node [     ] {} (g);
		\draw [->] (c) to node [     ] {} (f);
		\draw [->] (c) to node [     ] {} (g);
		\draw [->] (c) to node [     ] {} (h);
		\draw [dashed,->,thick] (d) to node [     ] {} (g);
		\draw [->] (d) to node [     ] {} (h);
		\draw [dashed,->,thick] (d) to node [     ] {} (i);
		\draw [dashed,->,thick] (e) to node [     ] {} (j);
		\draw [->] (e) to node [     ] {} (k);
		\draw [dashed,->,thick] (e) to node [     ] {} (l);
		\draw [->] (f) to node [     ] {} (k);
		\draw [->] (f) to node [     ] {} (l);
		\draw [->] (f) to node [     ] {} (m);
		\draw [dashed,->,thick] (g) to node [     ] {} (l);
		\draw [->] (g) to node [     ] {} (m);
		\draw [dashed,->,thick] (g) to node [     ] {} (n);
		\draw [->] (h) to node [     ] {} (m);
		\draw [->] (h) to node [     ] {} (n);
		\draw [->] (h) to node [     ] {} (o);
		\draw [dashed,->,thick] (i) to node [     ] {} (n);
		\draw [->] (i) to node [     ] {} (o);
		\draw [dashed,->,thick] (i) to node [     ] {} (p);
		
	\end{tikzpicture}
	}
	\caption{A three-step trinomial tree.}
	\label{fig}
\end{figure}
\subsection*{A trinomial tree approximation for the Black--Scholes model}
Let $h=T/n$, substituting $\omega_3 =-\omega_1 =\sqrt{3}$ and $\omega_2 = 0$  in Equation~\eqref{modified stratonovich} (from now on we type $S$ instead of $\hat{S}$), we introduce up, middle and down factors by
\begin{equation}\label{eq:umd}
	\begin{cases}
		u =\ln(u_0)=\ln\dfrac{S_3(h)}{S(0)}=(r-\dfrac{1}{2}\sigma^2)h + \sigma\sqrt{3h},\\
		m = \ln(m_0)=\ln\dfrac{S_2(h)}{S(0)}=(r-\dfrac{1}{2}\sigma^2)h,\\
		d = \ln(d_0)=\ln\dfrac{S_1(h)}{S(0)}=(r-\dfrac{1}{2}\sigma^2)h - \sigma\sqrt{3h}.
	\end{cases}
\end{equation}
Observe that the factors $u$, $m$, and $d$ depend on $n$ through $h$. To simplify notation, we omit this dependence.

\begin{proposition}
	The above trinomial construction produces a recombining trinomial tree.
\end{proposition}
\begin{proof}
	We simply calculate
	\begin{align*}
		u_0d_0=\exp \left\{2(r-\frac{1}{2}\sigma^2)h\right\} = m_0m_0=m_0^2.
	\end{align*}
Since $u_0,m_0$ and $d_0$ are positive, $m_0 =\sqrt{u_0d_0}$. Equivalently, $m=(u+d)/2$.
\end{proof}

As an example, assume that the stock price $S_0=\$100$, time to maturity $T=0.5$ of a year, strike price $K=\$120$, yearly interest rate $r=2.5\%$, yearly volatility  $\sigma=25\%$ and the number of steps in our trinomial tree $n=252$.

Now, given the above parameters, we write a script in MATLAB\textsuperscript{\textregistered},  where  we use the Black--Scholes formulae to calculate analytical prices, and apply the constructed trinomial tree model to estimate numerical prices of European call and put options.
The option prices are summarized in Table~\ref{tab:example}. Moreover, the behavior of trinomial tree prices (for call option) is depicted in Figure~\ref{fig:example}.
The absolute value of difference between  Black--Scholes and trinomial prices for call option is 0.0020704969 and for put is 0.00207049733.
\begin{figure}[t!]
	\centering
	\includegraphics[width=.8\linewidth]{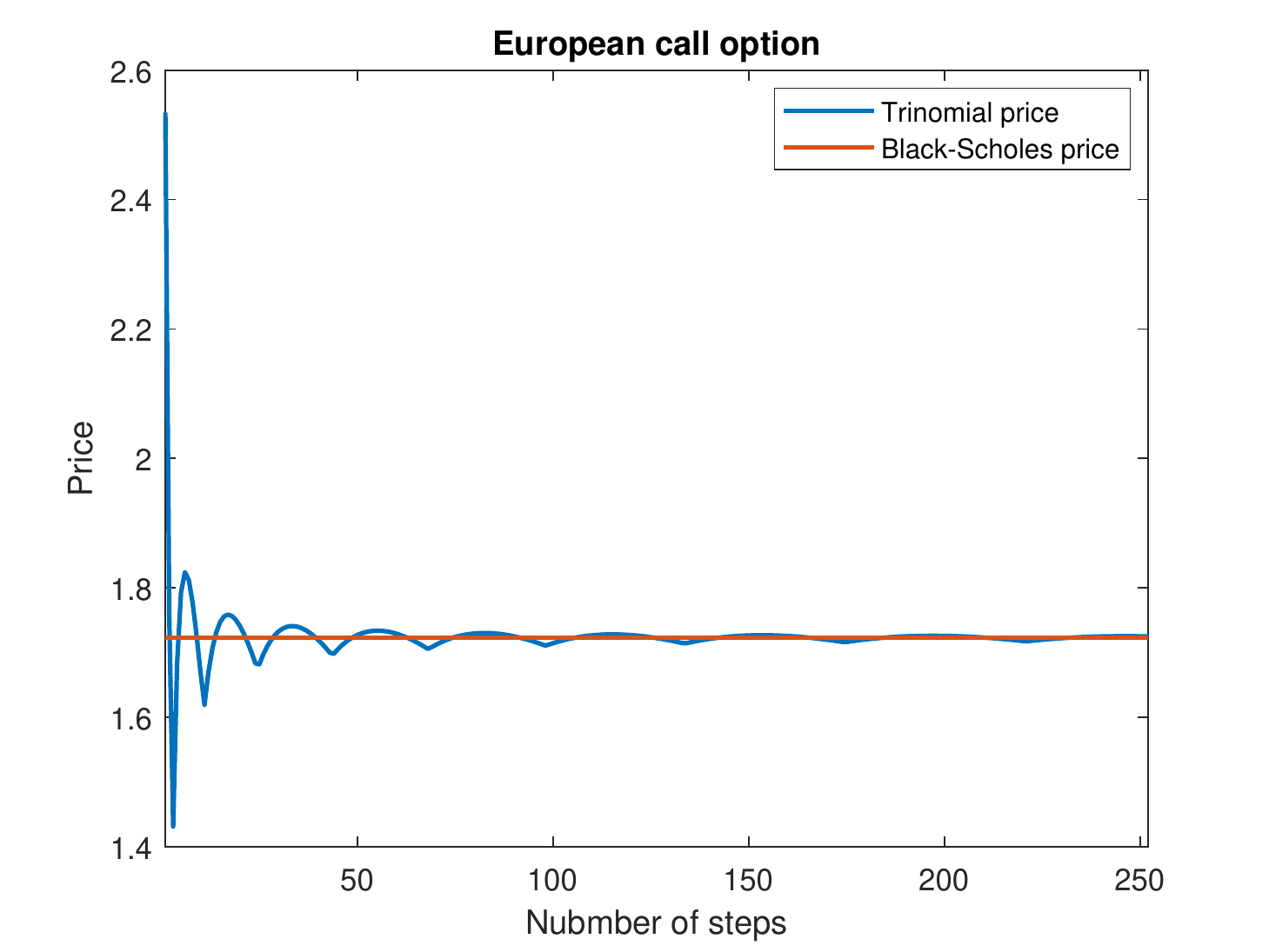}
	\caption{Black--Scholes and trinomial prices for European call  option.}
	\label{fig:example}
\end{figure}

\begin{table}[t!]
	\centering
	\resizebox{\linewidth}{!}{%
		\begin{tabular}{cccc}
			\toprule
			 \;$\pi_c$(\$) Black--Scholes \; & \;$\pi_c$(\$) Trinomial \; & \; $\pi_p$(\$) Black--Scholes \;&\; $\pi_p$(\$) Trinomial \; \\
			\midrule
			1.722901670   & 1.724972167  & 20.23223773 & 20.234308227 \\		
			\bottomrule
		\end{tabular}
	}
		\caption{Black--Scholes model example.}
		\label{tab:example}
\end{table}
\section{Convergence to geometric Brownian motion}
\label{Convergence to geometric Brownian motion}
As Figure~\ref{fig:example} suggests, our (numerical) trinomial price may converge to the  (analytical) Black--Scholes price. Now, we will consider a more general case and show that the sequence of our trinomial model (weakly) converges to a geometric Brownian motion. We first prove the convergence in distribution and then we prove the sequence of measures is  tight. Let us start by reviewing the necessary definitions given in~\cite{Billingsley}.

Let $(S,d)$ be a metric space. For the purposes of this paper it is enough to consider the case of $S=C[0,T]$ with the distance
\begin{align*}
d(a,b)=\max\limits_{t\in[0,T]}|a(t)-b(t)|.
\end{align*}

Let $\mathfrak{S}$ be the $\sigma$-field of Borel sets in $S$. For a probability space  $(\Omega,\mathfrak{F},\mathsf{P})$, consider a measurable map $X\colon\Omega\to S$. By this definition, for a point $w\in\Omega$, the image $X(w)$ is a continuous function on $[0,T]$. That is, $X(t,w)=X(w)(t)$ is a stochastic process with continuous sample paths.	

\begin{definition}
The \emph{distribution} of $X$ is the probability measure $P$ on the $\sigma$-field $\mathfrak{S}$ given by
\[
P(A)=\mathsf{P}(X^{-1}(A)),\qquad A\in\mathfrak{S}.
\]
\end{definition}

\begin{remark}
In other words, we describe the one-to-one correspondence between the family of stochastic processes $X$ with continuous trajectories and the family of probability measures $P$ on $C[0,T]$.	
\end{remark}

\begin{definition}

A sequence $\{\,X_n\colon n\geq 1\,\}$ of stochastic processes \emph{converges in distribution} to a stochastic process $X$ if the sequence $\{\,P_n\colon n\geq 1\,\}$ of their distributions weakly converges to the distribution $P$ of the stochastic process $X$, that is, for any continuous function $v\colon S\to\mathbb{R}$ we have
\[
\lim_{n\to\infty}\int_Sv(f)\,\mathrm{d}P_n(f)=\int_Sv(f)\,\mathrm{d}P(f).
\]
\end{definition}

Let $k$ be a positive integer, and let $t_1$, \dots, $t_k$ be arbitrary distinct points in $[0,T]$.

\begin{definition}
The \emph{natural projection} is the map $\pi_{t_1,\dots,t_k}\colon S\to\mathbb{R}^k$ given by
\[
\pi_{t_1,\dots,t_k}(f)=(f(t_1),\dots,f(t_k))^{\top}.
\]
\end{definition}

\begin{definition}
The \emph{finite-dimensional distributions} of a stochastic process $X$ are the measures $P\pi^{-1}_{t_1,\dots,t_k}$ on the Borel $\sigma$-field $\mathfrak{R}^k$ of the space $\mathbb{R}^k$ given by
\[
P\pi^{-1}_{t_1,\dots,t_k}(A)=P(\pi^{-1}_{t_1,\dots,t_k}(A)),\qquad A\in\mathfrak{R}^k.
\]
\end{definition}

\begin{definition}
A family $\Pi$ of probability measures on $\mathfrak{S}$ is called \emph{relatively compact} if every sequence of elements of $\Pi$ contains a weakly converging subsequence.
\end{definition}

\begin{theorem}
If the sequence $\{\,P_n\colon n\geq 1\,\}$ of the distributions of stochastic processes $\{\,X_n\colon n\geq 1\,\}$ is relatively compact and the finite-dimensional distributions of $X_n$ converge weakly to those of a stochastic process $X$, then $\{\,X_n\colon n\geq 1\,\}$ converges in distribution to $X$.
\end{theorem}

Thus, proof of convergence in distribution of the sequence of trinomial trees to the geometric Brownian motion is naturally dividing into two parts, which will be given in the next two subsections.

\subsection{Convergence of finite-dimensional distribution}
Let us study the convergence  of our model in the following 3 steps.
\paragraph{Step 1. (Preliminaries)} Consider Equation~\eqref{eq:umd} and  let $0=t_0<t_1<\cdots<t_n$, where $t_n = T$. Then, the price at each step of the tree can be found using
\begin{equation}
	\label{u1m1d1}
	S(t_i) =
	\begin{cases}
		S(t_{i-1})e^{u}\qquad \text{ with probability}\quad p_{u},\\
		S(t_{i-1})e^{m}\qquad \text{with probability} \quad p_{m},\qquad  i=1,\ldots,n\\
		S(t_{i-1})e^{d}\qquad \text{ with probability} \quad p_{d}.
	\end{cases}
\end{equation}

Denote the number of times that the price goes up, up to time $t_n$, i.e., step $n$, by $n_{u}(t_n)$, down by $n_{d}(t_n)$ and middle by $n_{m}(t_n) = n - n_{u}(t_n)-n_{d}(t_n)$, then the set of possible prices can be expressed by
\begin{align*}
	S(t_n)\in\left\{S_0\exp\left(un_{u}(t_n)+dn_{d}(t_n)+m[n-n_{u}(t_n)-n_{d}(t_n)]\right)\right\}.
\end{align*}

If we substitute $m=(u+d)/2$, the above set reduces to
\begin{align*}
	S(t_n)\in\left\{S_0\exp\left([n_{u}(t_n)-n_{d}(t_n)](u-d)/2+n(u+d)/2\right)\right\}.
\end{align*}

Put $\Omega=C[0,T]$, the space of paths, and $\mathfrak{S}$ as the $\sigma$-field of Borel sets.  Let $p=(p_u,p_m,p_d)$ such that $0\leq p_u,p_m,p_d\leq 1$ and $p_u+p_m+p_d=1$ be given. Now, we define the probability measure $\mathsf{P}_n$ on $\mathfrak{S}$ supported by a finite subset of sample space $C[0,T]$ with $3^n$ elements.

\paragraph{Step 2. (Description of the measure $\mathsf{P}_n$)}
Consider the set of $3^n$ ``words'' $w=w_1\ldots w_n$ consisting of $n$ letters, where each letter is either $u$ or $m$ or $d$. Let
\begin{itemize}
	\item $\{\,\xi_k\colon k\geq 1\,\}$ be the sequence of independent and identically distributed (IDD) random variables with $\mathsf{P}\{\xi_1=\ln(w_1)\}=p_{w_1}$,
	\item $\Xi_0=0$ and $\Xi_k=\xi_1+\cdots+\xi_k$, $k\geq 1$,
	\item $f_{w}\in C[0,T]$ be the function that takes the following values:
	\begin{enumerate}		
		\item	\[
		f_{w}(0)=\ln(S_0),
		\]
		\item 	\[
		f_{w}\left(\dfrac{(k+1)T}{n}\right)=f_{w}\left(\dfrac{kT}{n}\right)+\xi_{k+1},\qquad 0\leq k\leq n-1,
		\]
		\item	 the value in an arbitrary ``intermediate'' point, say $t$, is given by linear interpolation, that is
		\[
		f_{w}(t)=\ln(S_0)+\Xi_{\lfloor nt/T\rfloor}+(nt/T-\lfloor nt/T\rfloor)\xi_{\lfloor nt/T\rfloor+1},\qquad 0\leq t\leq T.
		\]	
	\end{enumerate}
	\item $n_{u}(w)$ represents the number of letters $u$,
	\item $n_{m}(w)$ represents the number of letters $m$,
	\item $n_{d}(w)$ represents the number of letters $d$.
\end{itemize}

Now, we define  the probability measure $\mathsf{P}_n$ on the Borel $\sigma$-field supported by the subset of sample space $C[0,T]$ with $3^n$ atoms $f_{w}(t)$ by

\[
\mathsf{P}_n(f_{w}(t))=p^{n_u(w)}_{u}p^{n_m(w)}_{m}p^{n_d(w)}_{d}.
\]

Finally, let $\mathsf{P}$ be the measure that corresponds to the geometric Brownian motion
\[
S(t)=S_0\exp\left\{(r-\frac{1}{2}\sigma^2)t+\sigma W(t)\right\},\qquad 0\leq t\leq T,
\]
that is, $\mathsf{P}(A)=\mathsf{P}\{S(t)\in A\}$, $\;A\in\mathfrak{S}$.

\paragraph{Step 3. (Proof of convergence)}
Define the sequence of random variables $\{X_i\}_{i=1}^n$ on the probability space $(\Omega,\mathfrak{F},\mathsf{P}_n,(\mathfrak{F}_t)_{t\geq 0})$ and on $w=(\xi_1,\ldots,\xi_n)\in\Omega$  as $X(w)=\xi_i$. Following~\cite{Pascucci2011}, we define
\begin{align*}
	X_i(w) =
	\begin{cases}
		\;\;\; 1 \qquad \text{if} \quad \xi_i=u,\\
		\;\;\; 0 \qquad \text{if} \quad \xi_i=m,\\
		 -1\qquad \text{if}\quad \xi_i=d.		
	\end{cases}
\end{align*}
Note that, the sequence $\{X_i\}_{i=1}^{n}$ is a sequence of independent and identically distributed (IID) random variables.

Let  $S^n(t)$ be the stochastic process that corresponds to the measure $\mathsf{P}_n$. By the constructions in ``\emph{Step 2.}'', the values of the process $S^n(t_i)$ (for simplicity denote it by $S(t_i)$) in our trinomial tree can be expressed by re-writing
Equation~\eqref{u1m1d1}  as the following trinomial stock price
\begin{align}
	\label{gen1}
	S(t_i) = S(t_{i-1})\exp\left\{\frac{(u+d)}{2}+\frac{(u-d)}{2}X_i\right\},\qquad i=1,\ldots,n.
\end{align}

We observe that $\mathsf{E}[X_i]=p_{u}(1)+p_{m}(0)+p_{d}(-1)=p_{u}-p_{d}$
and $ \mathsf{E}[X_i^2]=p_{u}(1)^2+p_{m}(0)^2+p_{d}(-1)^2=p_{u}+p_{d}$. Using the equation above, we obtain following mean
\begin{align}
	\label{mu}
 \mathsf{E}\left[\ln\frac{S(t_i)}{S(t_{i-1})}\right]&=\mathsf{E}\left[\frac{(u+d)}{2}+\frac{(u-d)}{2}X_i\right]=\frac{(u+d)}{2}+\frac{(u-d)}{2}\mathsf{E}[X_i],\nonumber\\
&=	\frac{1}{2}[u+d+(u-d)(p_{u}-p_{d})],
\end{align}
and variance
\begin{align*}
	 \text{Var}\left[\ln\frac{S(t_i)}{S(t_{i-1})}\right]&=\mathsf{E}\left[\left(\ln\frac{S(t_i)}{S(t_{i-1})}\right)^2\right]-\left(\mathsf{E}\left[\ln\frac{S(t_i)}{S(t_{i-1})}\right]\right)^2,
\end{align*}
where for the first term on the right hand side of the above equation, we have
\begin{align*}
	\mathsf{E}\left[.^2\right]&=\dfrac{1}{4}\mathsf{E}\left[(u+d)^2+2(u+d)(u-d)X_i+(u-d)^2X_i^2\right]\\
	&=\dfrac{1}{4}\left((u+d)^2+2(u+d)(u-d)\mathsf{E}\left[X_i\right]+(u-d)^2\mathsf{E}\left[X_i^2\right]\right)\\
	&=\dfrac{1}{4}\left((u+d)^2+2(u+d)(u-d)(p_{u}-p_{d})+(u-d)^2(p_{u}+p_{d})\right),
\end{align*}
and for the second term
\begin{align*}
	\left(\mathsf{E}\left[.\right]\right)^2&=\dfrac{1}{4}[(u+d)^2+2(u+d)(u-d)(p_{u}-p_{d})+(u-d)^2(p_{u}-p_{d})^2.
\end{align*}
Subtracting  second term from first one gives
\begin{align*}
	\mathsf{E}\left[.^2\right]-\left(\mathsf{E}\left[.\right]\right)^2&=\dfrac{1}{4}\left[(u-d)^2(p_{u}+p_{d})-(u-d)^2(p_{u}-p_{d})^2\right]\\
	&=\dfrac{1}{4}[(p_{u}+p_{d})-(p_{u}-p_{d})^2](u-d)^2.
\end{align*}
Thus,
\begin{align}
	\label{sigma}
	\text{Var}\left[\ln\frac{S(t_i)}{S(t_{i-1})}\right]=\frac{1}{4}[p_{u}+p_{d}-(p_{u}-p_{d})^2](u-d)^2.
\end{align}

Substituting the values  $\lambda_1=p_{u}=p_{d}=\lambda_3=1/6$, $u$ and $d$, confirms the construction of our trinomial tree. That is,
\begin{align*}
	\mathsf{E}\left[\ln\frac{S(t_i)}{S(t_{i-1})}\right]=\frac{1}{2}[u+d+(p_{u}-p_{d})(u-d)]=\frac{1}{2}(\mu h+\mu h+0\cdot (u-d))=\mu h,
\end{align*}
where $\mu=(r-\frac{1}{2}\sigma^2)$. Moreover,
\begin{align*}
	\text{Var}\left[\ln\frac{S(t_i)}{S(t_{i-1})}\right]=\frac{1}{4}[p_{u}+p_{d}-(p_{u}-p_{d})^2](u-d)^2=\frac{1}{4}\frac{2}{6}(2\sigma\sqrt{3h})^2= \sigma^2h.
\end{align*}

In the next step, we would like to prove that the sequence of finite-dimensional distributions of our trinomial trees converges to those of the geometric Brownian motion. In one dimensional case, we prove $S(t_n)$ in our trinomial model converges in distribution to $S(T)$ in geometric Brownian motion as $n\to\infty$. First, we
iterate Equation~\eqref{gen1}. This yields to
\begin{align}
	\label{eq:iterated}
	S(t_n) = S_0\exp\left\{n\frac{(u+d)}{2}+\frac{(u-d)}{2}\sum_{i=1}^nX_i\right\}.
\end{align}

Note that for $i=1,\ldots,n$, the random variables $\ln\dfrac{S(t_{i})}{S(t_{i-1})}$ are IID random variables. Now, we use Equation~\eqref{gen1} and define $Z_n$ as the standardized random variable of the sum, $\sum\limits_{i=1}^n\ln\dfrac{S(t_{i})}{S(t_{i-1})}\left(=\ln\dfrac{S(t_{n})}{S_0}\right)$. That is,
\[
Z_n\coloneqq\dfrac{\sum\limits_{i=1}^n\ln\dfrac{S(t_{i})}{S(t_{i-1})}-n\mathsf{E}\left[\ln\dfrac{S(t_{i})}{S(t_{i-1})}\right]}{\left(n\text{Var}\left[\ln\dfrac{S(t_{i})}{S(t_{i-1})}\right]\right)^{1/2}}.
\]
Using the  obtained mean and variance values in Equations~\eqref{mu} and~\eqref{sigma}
\begin{gather*}
	Z_n
	=\dfrac{\sum\limits_{i=1}^n\ln\dfrac{S(t_{i})}{S(t_{i-1})}-\dfrac{n}{2}[u+d+(p_{u}-p_{d})(u-d)]}{\sqrt{\dfrac{n}{4}[p_{u}+p_{d}-(p_{u}-p_{d})^2](u-d)^2}}\\
	=\dfrac{\dfrac{(u-d)}{2}\sum_{i=1}^nX_i+n\dfrac{u+d}{2}-n\dfrac{u+d}{2}+n\dfrac{(p_{u}-p_{d})(u-d)]}{2}}{\sqrt{n}\sqrt{\dfrac{1}{4}[p_{u}+p_{d}-(p_{u}-p_{d})^2](u-d)^2}},
\end{gather*}
multiplying numerator and denominator by $2/(u-d)$,
\begin{align*}	
	Z_n=\dfrac{\sum_{i=1}^nX_i-n(p_{u}-p_{d})}{\sqrt{n}\sqrt{p_{u}+p_{d}-(p_{u}-p_{d})^2}}
	=\dfrac{X_1+\ldots+X_n-n(p_{u}-p_{d})}{\sqrt{n}\sqrt{p_{u}+p_{d}-(p_{u}-p_{d})^2}},
\end{align*}
where for all $X_i$ and $i=1,\ldots,n$
\begin{align*}
	\mathsf{E}[X_i]&=p_{u}(1)+p_{m}(0)+p_{d}(-1)=p_{u}-p_{d}=\mu_{X},\\
	\text{Var}[X_i]&=\mathsf{E}[X_i^2]-(\mathsf{E}[X_i])^2=p_{u}(1)^2+p_{m}(0)^2+p_{d}(-1)^2-(p_{u}-p_{d})^2\\
	&=p_{u}+p_{d}-(p_{u}-p_{d})^2=\sigma^2_{X}.
\end{align*}
Thus,
\begin{align*}
	Z_n=\dfrac{X_1+\ldots+X_n-n\mu_{X}}{\sigma_{X}\sqrt{n}},
\end{align*}
and using central limit theorem  (see~\cite{Kijima2013}), we have
\begin{align*}
	\lim_{n\to\infty}\mathsf{P}\{Z_n\leq z\}=\mathcal{N}(z),\qquad z\in \mathbb{R},
\end{align*}
where $\mathcal{N}(z)$ is the standard normal distribution function.

Doing a little algebra, we calculate the following terms in Equation~\eqref{eq:iterated}
\begin{align}
	\label{S(T)_GBM}
	n\left(\frac{u+d}{2}\right) &= (r-\frac{1}{2}\sigma^2)nh=(r-\frac{1}{2}\sigma^2)T,\nonumber \\
	\frac{u-d}{2}\sum_{i=1}^nX_i &= \frac{u-d}{2}\left[n(p_u-p_d)+Z_n\sqrt{n}\sqrt{p_{u}+p_{d}-(p_{u}-p_{d})^2}\right]\nonumber\\
	& = {(\sigma\sqrt{3h})}(Z_n\sqrt{n}\sqrt{{1}/{3}})]=\sigma\sqrt{nh}Z_n=\sigma\sqrt{T}Z_n,
\end{align}
where, $p_u=p_d=1/6$, $T=t_n=nh$, and $u,d$ are given in Equation~\eqref{eq:umd}. Substituting the above values into Equation~\eqref{eq:iterated} (for an arbitrary $T$), we get
\begin{align*}
	S(T)=S_0\exp\left\{(r-\frac{1}{2}\sigma^2)T+\sigma\sqrt{T}Z_n\right\}.
\end{align*}

As the result, we showed that the one-dimensional distributions of our trinomial tree model  converge (in distribution) to those of the geometric Brownian motion as $n\to \infty$. The generalization to $k$-dimensional distributions is straightforward and uses the fact that the stochastic processes $X_n$ and $X$ have independent increments, see \cite{Billingsley} for details.
It remains to prove that the corresponding sequence of measures is relatively compact.
\subsection{Relative compactness}
To prove the  relative compactness of the sequence of measures, we firstly denote the stock price, up, middle and down factors as  functions of $n$. That is, $S^n,u_n,m_n$ and $d_n$. Secondly, we use the \emph{modulus of continuity}'s definition given in~\cite[Equation~7.1]{Billingsley} and Theorem~\cite[Theorem~7.5]{Billingsley}. 

\begin{definition}[Modulus of continuity]
\label{Def:1}
The \emph{modulus of continuity} for (an arbitrary) function $S^n(\cdot)$ on $[0,T]$	is defined by
\begin{align*}
	g(S^n,\delta)=\sup\limits_{\lvert z-z_0\rvert\leq \delta}\lvert S^n(z)-S^n(z_0)\rvert.
\end{align*}
\end{definition}

\begin{theorem}[Tightness and compactness in $C$]
\label{Thm:1}
Let $S,S^1,S^2,\ldots$ be random functions. If for all $t_1,\ldots,t_k$
\begin{align*}
(\mathcal{C}_1)\qquad	(S^n_{t_1},\ldots,S^n_{t_k})\Rightarrow_n (S_{t_1},\ldots,S_{t_k}),
\end{align*}
holds and if for each positive $\epsilon$
\begin{align*}
(\mathcal{C}_2)\qquad	\lim_{\delta\downarrow 0}\lim_{n\to\infty}\sup \mathsf{P}\left\{g(S^n,\delta)\geq \epsilon\right\}=0,
\end{align*}
holds, then $S^n \Rightarrow_n S$.
\end{theorem}

On the one hand, we have already shown that the sequence of our trinomial models converges (in distribution) to the geometric Brownian motion. Therefore, condition $(\mathcal{C}_1)$ holds.
On the other hand,  using Definition~\ref{Def:1} and since $\lvert S^n(z)-S^n(z_0)\rvert$ are piece-wise linear, we have
\begin{align*}
	g(S^n,\delta)=\sup\limits_{\lvert z-z_0\rvert\leq \delta}\lvert S^n(z)-S^n(z_0)\rvert=\max\{e^{u_n},e^{m_n},e^{d_n}\}\delta.
\end{align*}

Let $C_n=\{e^{u_n},e^{m_n},e^{d_n}\}$, then $\lim\limits_{n\to\infty}\max \{C_n\}=1$. Thus, for condition $(\mathcal{C}_2)$ in Theorem~\ref{Thm:1}, we have
\begin{align*}
&\lim_{\delta\downarrow 0}\lim_{n\to\infty}\sup \mathsf{P}\left\{g(S^n,\delta)\geq \epsilon\right\}\\
=&\lim_{\delta\downarrow 0}\lim_{n\to\infty}\sup \mathsf{P}\left\{\max(\{C_n\})\delta\geq \epsilon\right\}\\
=&\lim_{\delta\downarrow 0} \mathsf{P}\left\{\delta\geq \epsilon\right\}=0.
\end{align*}
\section{Martingale probability measure}
\label{Martingale probability measure}
In this section, we would like to investigate  if there exists a \emph{martingale (risk-neutral) probability measure} $\mathsf{Q}$ equivalent  to the physical probability measure $\mathsf{P}$ such that the discounted price process $S(t_n)$ becomes a martingale.

On the one hand, using the fundamental asset pricing theorem,   \emph{no arbitrage opportunity} is possible if and only if there exists a risk-neutral probability measure (see~\cite{Kijima2013}). This means that the average return on an asset should be equal to risk-free return, i.e.,
\begin{enumerate}[label=($\mathcal{NA}$),wide=0pt]
	\item $\mathsf{E}^\mathsf{Q}[S(t_i)\vert S(t_{i-1})]=e^{r(t_i-t_{i-1})}S(t_{i-1})=e^{rh}S(t_{i-1})$.	
\end{enumerate}

On the other hand, by definition the martingale conditions for an arbitrary discrete time stochastic process $\{X(t)\}_{t\geq 0}$ are:
\begin{enumerate}[label=($\mathcal{M}_\arabic*$), wide=0pt]
	\item $\mathsf{E}^\mathsf{Q}[\lvert X(t_n)\rvert]<\infty,$
	\item $\mathsf{E}^\mathsf{Q}[X(t_i)\vert X(t_{i-1})]=X(t_{i-1})$.
\end{enumerate}

In order to see under which conditions the \emph{discounted} asset price process $S(t)$ in a risk-neutral world is a martingale, we first define the stochastic process $\{Y_t\}_{t\in[0,T]}$ by $Y(t_n)=\prod_{i=1}^ny_i$ on $w=(\xi_1,\ldots,\xi_n)\in\Omega$ as $y(w)=\xi_i$, where
\begin{align*}
	y_i(w) =
	\begin{cases}
		\;u_0 \qquad \text{if} \quad \xi_i=u,\\
		m_0 \qquad \text{if} \quad \xi_i=m,\\
		\;d_0 \qquad \text{if} \quad \xi_i=d.		
	\end{cases}
\end{align*}
Now, we can write the  stock price ${S(t_n)}$ as
\[
S(t_n)=S_0Y(t_n).
\]

Firstly, for condition ($\mathcal{M}_2$), we have
\begin{align*}
	\mathsf{E}^\mathsf{Q}[S(t_i)\vert S(t_{i-1})]&=\mathsf{E}^\mathsf{Q}[S(t_{i-1})Y(t_{i})\vert S(t_{i-1})]
	=S(t_{i-1})\mathsf{E}^\mathsf{Q}[y_i]\\
	&=S(t_{i-1})(p_{u}{u_0}+p_m{m_0}+p_d{d_0}).
\end{align*}
Equating the above equation with the no arbitrage condition $(\mathcal{NA})$ yields to
\begin{equation}
	\label{Martingale condition}
	e^{rh}=p_{u}{u_0}+p_m{m_0}+p_d{d_0}=p_{u}e^{u}+p_me^{\frac{u+d}{2}}+p_de^{d}.
\end{equation}

Secondly, for condition ($\mathcal{M}_1$) and  since $S(t_n), S_0>0$, we have
\begin{align*}
	\mathsf{E}^\mathsf{Q}[\lvert S(t_n)\rvert]&=\mathsf{E}^\mathsf{Q}[S(t_n)]=\mathsf{E}^\mathsf{Q}[S_0Y(t_n)]=S_0\prod_{i=1}^{n}\mathsf{E}^\mathsf{Q}[y_i]\\
	&=S_0(p_{u}{u_0}+p_m{m_0}+p_d{d_0})^n=S_0(e^{rh})^n=S_0e^{rt_n}.
\end{align*}
Consequently, the expected value of \emph{discounted} risky asset price becomes
\begin{align*}
	\mathsf{E}^\mathsf{Q}[\lvert e^{-rt_n} S(t_n)\rvert]&=e^{-rt_n}S_0e^{rt_n} = S_0<\infty.
\end{align*}

Thus, the discounted stock price is a martingale if Equation~\eqref{Martingale condition}  holds.  Moreover, we note that if $h\to 0$, then both hand sides of Equation~\eqref{Martingale condition} tend to 1. We  investigate this condition in an example given in the next section.

\section{Extension of the results and examples}
\label{Extension of the results and examples}
The proof of convergence suggests us that there might be many solutions to a recombining trinomial tree.

We extend Equation~\eqref{eq:umd} by introducing a finite positive parameter $c>1$, that is
\begin{equation*}
	\begin{cases}
		u = (r-\frac{1}{2}\sigma^2)h + \sigma\sqrt{ch},\\
		m = (r-\frac{1}{2}\sigma^2)h,\\
		d = (r-\frac{1}{2}\sigma^2)h - \sigma\sqrt{ch}.
	\end{cases}
\end{equation*}
Also, we put $p_{u}=p_{d}=1/(2c)$ and $p_{m}=1-2p=1-1/c$.
Note that, with the above general formulations our proof of convergence remains to be valid. Indeed, Equation~\eqref{S(T)_GBM} holds for different values of $c$.  In our trinomial tree, $c=3$, $p_{u}=p_{d}=1/6=1/(2c)$ and $p_{m}=1-2p=1-1/c$. We would like to examine the convergence of Black--Scholes (BS) price and extension of our trinomial (Tri.) model price where $c=\{1,3/2,2,3,4,5,10,20,30\}$.

Let the current stock price $S_0=\$100$, time to maturity $T=1$ year, yearly interest rate $r=3.5\%$ and yearly volatility  $\sigma=30\%$ be given. We set (number of steps in the tree) $n=252$ (business days in 1 year). Now, we calculate the price of European call ($\pi_c$) and put ($\pi_p$) options with the described parameters for strike prices $K=\{\$80,\$100,\$120\}$. We calculate the absolute error (Abs. error) by  absolute value of the difference between analytical Black--Scholes (BS) price and trinomial price (Tri.).  Using MATLAB\textsuperscript{\textregistered}, the results are given in Table~\ref{tab2}, \ref{tab3} and \ref{tab4}.
\begin{table}[t!]
	\centering
	\resizebox{\linewidth}{!}{%
		\begin{tabular}{cccccccc}
			\toprule
			$c$ \;& $p_{u}$ \;& $\pi_c$(\$) BS \;& $\pi_c$(\$) Tri. \;& Abs. error $\pi_c$ \;& $\pi_p$(\$) BS  \;& $\pi_p$(\$) Tri.  \;& Abs. error $\pi_p$  \\
			\midrule
			1   & 1/2  & 25.578 & 25.583 & 0.0050237 & 2.8262 & 2.8315 & 0.0052915 \\
			3/2 & 1/3  & 25.578 & 25.579 & 0.0008575 & 2.8262 & 2.8272 & 0.0010584 \\
			2   & 1/4  & 25.578 & 25.574 & 0.0034202 & 2.8262 & 2.8229 & 0.0032862 \\
			3   & 1/6  & 25.578 & 25.581 & 0.0035653 & 2.8262 & 2.8297 & 0.0035653 \\
			4   & 1/8  & 25.578 & 25.581 & 0.0031566 & 2.8262 & 2.8292 & 0.0030227 \\
			5   & 1/10 & 25.578 & 25.568 & 0.0102440 & 2.8262 & 2.8157 & 0.0105120 \\
			10  & 1/20 & 25.578 & 25.585 & 0.0071857 & 2.8262 & 2.8324 & 0.0062482 \\
			20  & 1/40 & 25.578 & 25.591 & 0.0134450 & 2.8262 & 2.8374 & 0.0111680 \\
			30  & 1/60 & 25.578 & 25.511 & 0.0662680 & 2.8262 & 2.7563 & 0.0698850 \\
			\bottomrule
		\end{tabular}
	}
		\caption{$K=\$80$.}
		\label{tab2}
\end{table}

\begin{table}[t!]
	\centering
	\resizebox{\linewidth}{!}{%
		\begin{tabular}{cccccccc}
		\toprule
		$c$ \;& $p_{u}$ \;& $\pi_c$(\$) BS \;& $\pi_c$(\$) Tri. \;& Abs. error $\pi_c$ \;& $\pi_p$(\$) BS  \;& $\pi_p$(\$) Tri.  \;& Abs. error $\pi_p$  \\
		\midrule
		1   & 1/2  & 13.517 & 13.523 & 0.0058724 & 10.078 & 10.084 & 0.00614020 \\
		3/2 & 1/3  & 13.517 & 13.522 & 0.0051641 & 10.078 & 10.083 & 0.00536490 \\
		2   & 1/4  & 13.517 & 13.522 & 0.0047397 & 10.078 & 10.083 & 0.00487370 \\
		3   & 1/6  & 13.517 & 13.520 & 0.0031506 & 10.078 & 10.081 & 0.00315060 \\
		4   & 1/8  & 13.517 & 13.518 & 0.0009543 & 10.078 & 10.079 & 0.00082035 \\
		5   & 1/10 & 13.517 & 13.516 & 0.0016326 & 10.078 & 10.076 & 0.00190050 \\
		10  & 1/20 & 13.517 & 13.500 & 0.0177320 & 10.078 & 10.059 & 0.01867000 \\
		20  & 1/40 & 13.517 & 13.460 & 0.0570400 & 10.078 & 10.018 & 0.05931700 \\
		30  & 1/60 & 13.517 & 13.416 & 0.1008700 & 10.078 & 9.9733 & 0.10449000 \\
		\bottomrule
		\end{tabular}
	}
		\caption{$K=\$100$.}
		\label{tab3}
\end{table}

\begin{table}[t!]
	\centering
	\resizebox{\linewidth}{!}{%
	\begin{tabular}{cccccccc}
		\toprule
		$c$ \;& $p_{u}$ \;& $\pi_c$(\$) BS \;& $\pi_c$(\$) Tri. \;& Abs. error $\pi_c$ \;& $\pi_p$(\$) BS  \;& $\pi_p$(\$) Tri.  \;& Abs. error $\pi_p$  \\
		\midrule
		1   & 1/2  & 6.4401 & 6.4333 & 0.006796 & 22.313 & 22.306 & 0.00652820 \\
		3/2 & 1/3  & 6.4401 & 6.4424 & 0.002330 & 22.313 & 22.315 & 0.00253110 \\
		2   & 1/4  & 6.4401 & 6.4401 & 0.000055 & 22.313 & 22.313 & 0.00018848 \\
		3   & 1/6  & 6.4401 & 6.4363 & 0.003782 & 22.313 & 22.309 & 0.00378220 \\
		4   & 1/8  & 6.4401 & 6.4317 & 0.008373 & 22.313 & 22.304 & 0.00850690 \\
		5   & 1/10 & 6.4401 & 6.4481 & 0.008076 & 22.313 & 22.321 & 0.00780800 \\
		10  & 1/20 & 6.4401 & 6.4366 & 0.003491 & 22.313 & 22.308 & 0.00442890 \\
		20  & 1/40 & 6.4401 & 6.4442 & 0.004160 & 22.313 & 22.315 & 0.00188270 \\
		30  & 1/60 & 6.4401 & 6.3995 & 0.040562 & 22.313 & 22.269 & 0.04417800 \\
		\bottomrule
	\end{tabular}
	}
	\caption{$K=\$120$.}
	\label{tab4}
\end{table}

Moreover, we investigate if the martingale condition holds for this example. This means, the following absolute value of difference (call it martingale condition) as $n\to\infty$ or equivalently $h\to 0$ should tend to zero.
\[
\text{Martingale condition}=\lvert(p_{u}u_0+p_mm_0+p_dd_0)-e^{rh}\rvert\xrightarrow[h\to 0]{}0.
\]
The martingale conditions for different values of $c$  are illustrated in Table~\ref{tab5}.
\begin{remark}
We note that, when $c=1$, our trinomial model reduces to a binomial model where $p_u=p_d=1/2$ (see the dashed lines in Figure~\ref{fig}). In the limit case, i.e., when $h\to 0$, our reduced trinomial model matches the Jarrow--Rudd binomial model proposed in~\cite{Jarrow-Rudd} (see also,~\cite{Jarrow}). 	
\end{remark}

\begin{table}[t!]
	\centering
	\begin{tabular}{ccc}
		\toprule
		$c$ &\qquad $p_{u}$ \qquad&  Martingale condition  \\
		\midrule
		1   &\qquad 1/2  \qquad&  $1.0630\times 10^{-8}$  \\
		3/2 &\qquad 1/3  \qquad&  $7.9724\times 10^{-9}$  \\
		2   &\qquad 1/4  \qquad&  $5.3151\times 10^{-9}$  \\
		3   &\qquad 1/6  \qquad&  $3.7947\times 10^{-13}$  \\
		4   &\qquad 1/8  \qquad&  $5.3145\times 10^{-9}$  \\
		5   &\qquad 1/10 \qquad&  $1.0629\times 10^{-8}$  \\
		10  &\qquad 1/20 \qquad&  $3.7206\times 10^{-8}$  \\
		20  &\qquad 1/40 \qquad&  $9.0369\times 10^{-8}$  \\
		30  &\qquad 1/60 \qquad&  $1.4355\times 10^{-7}$  \\
		\bottomrule
		\end{tabular}
		\caption{Martingale condition.}
		\label{tab5}
\end{table}

\subsection*{A trinomial tree approximation for Black's model}

So far, we have considered the Black--Scholes model. Constructing trinomial model and a recombining tree in the Black's model is similar. Proofs of convergence to the geometric Brownian motion and martingale conditions can be achieved  similar to what we have done in previous sections. We only have different up, middle and down factors. That is, using Equation~\eqref{modified stratonovich for forward rate}
\begin{equation*}
	\begin{cases}
		u =\ln(u_0)=\ln \dfrac{F_3(h)}{F(0)}=-\dfrac{1}{2}\sigma^2h + \sigma\sqrt{3h},\\
		m = \ln(m_0)=\ln\dfrac{F_2(h)}{F(0)} =-\dfrac{1}{2}\sigma^2h,\\
		d = \ln(d_0)=\ln\dfrac{F_1(h)}{F(0)}=-\dfrac{1}{2}\sigma^2h - \sigma\sqrt{3h}.
	\end{cases}
\end{equation*}
Now, we try an example inspired by~\cite{Hull2017}. Assume that the current future price of a commodity  $F_0=\$100$, time to maturity $T=0.5$ of a year, strike price $K=\$120$, yearly interest rate $r=2.5\%$, yearly volatility  $\sigma=25\%$ and the number of steps in our trinomial tree $n=252$.  Programming in MATLAB\textsuperscript{\textregistered}, we find the prices of European call futures and put futures options  given in Table~\ref{tab:example2}. Moreover, the behaviour of trinomial trees (for put futures option) as the number of steps increases is depicted in Figure~\ref{fig:example2}.
\begin{figure}[t!]
	\centering
	\includegraphics[width=.8\linewidth]{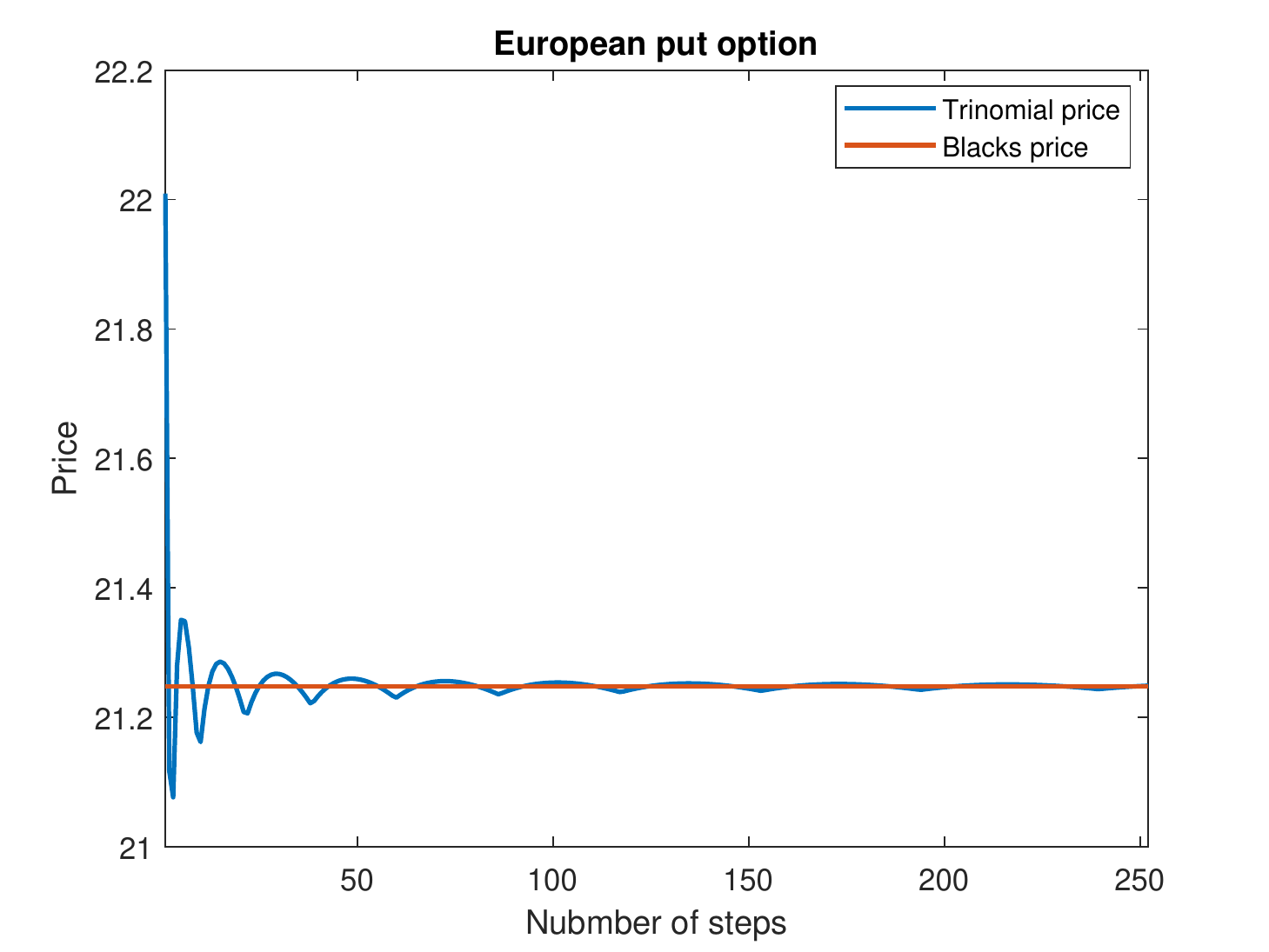}
	\caption{Black's and trinomial prices for European  put futures option.}
	\label{fig:example2}
\end{figure}

\begin{table}[t!]
	\centering
		\begin{tabular}{cccc}
			\toprule
			$\pi_c$(\$) Black \;& $\pi_c$(\$) Trinomial \;& $\pi_p$(\$) Black \;& $\pi_p$(\$) Trinomial    \\
			\midrule
			 1.496683230   & 1.497311844  & 21.248239239 & 21.248867854 \\		
			\bottomrule
		\end{tabular}
		\caption{Black's model example.}
		\label{tab:example2}
\end{table}
The  absolute value of difference between  Black's and trinomial prices for call futures option is 0.0006286140 and for put is 0.0006286144.
\subsection{Pricing American options using the trinomial tree}
In this part, we will investigate the performance of our model in pricing American call and put options and compare it with classical CRR model. Let the current stock price $S_0=\$100$, time to maturity $T=0.5$ year, yearly interest rate $r=2.5\%$ and yearly volatility  $\sigma=25\%$ be given. We increase the number of steps in the tree up to $n=126$ (business days in half a year). We would like to calculate the price of American call and put  options with the described parameters for strike prices $K=\{\$90,\$100,\$110\}$. The stock is non-dividend paying, thus the optimal time to exercise such an American call option is at maturity and therefore its price should coincide with a European call option. This fact help us to calculate the Black--Scholes (BS) price of European call option and calculate absolute errors of our trinomial model and CRR model for a call option.

Using MATLAB\textsuperscript{\textregistered}, we implement \emph{backward recursion} to preform these calculations. The results are given in Figure~\ref{fig:AmExK90}, \ref{fig:AmExK100}, \ref{fig:AmExK110} and \ref{fig:abserror}.
\begin{figure}[t!]
	\begin{subfigure}{.5\linewidth}
		\centering
		\includegraphics[width=\linewidth]{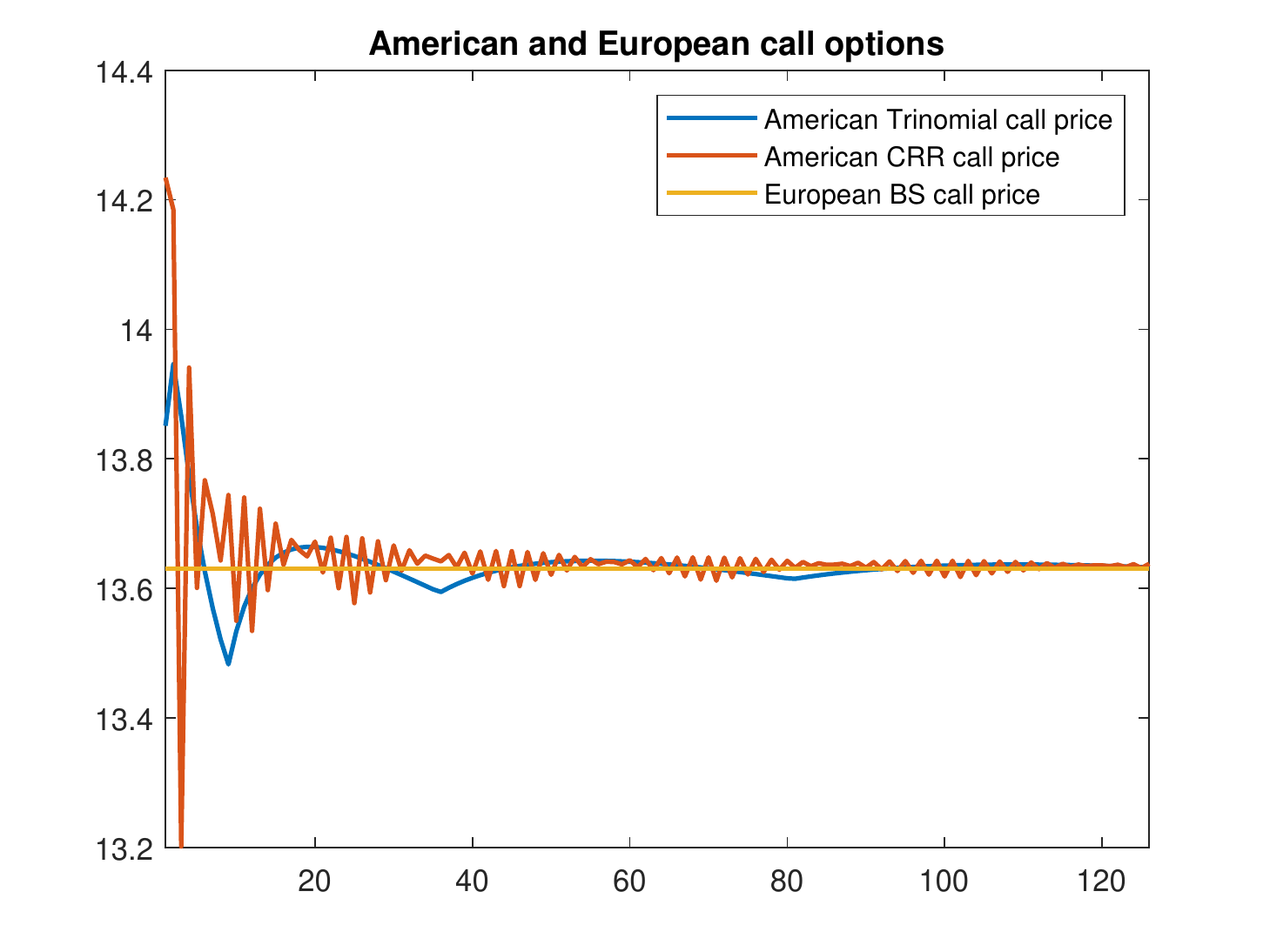}
	\end{subfigure}
	\begin{subfigure}{.5\linewidth}
		\centering
		\includegraphics[width=\linewidth]{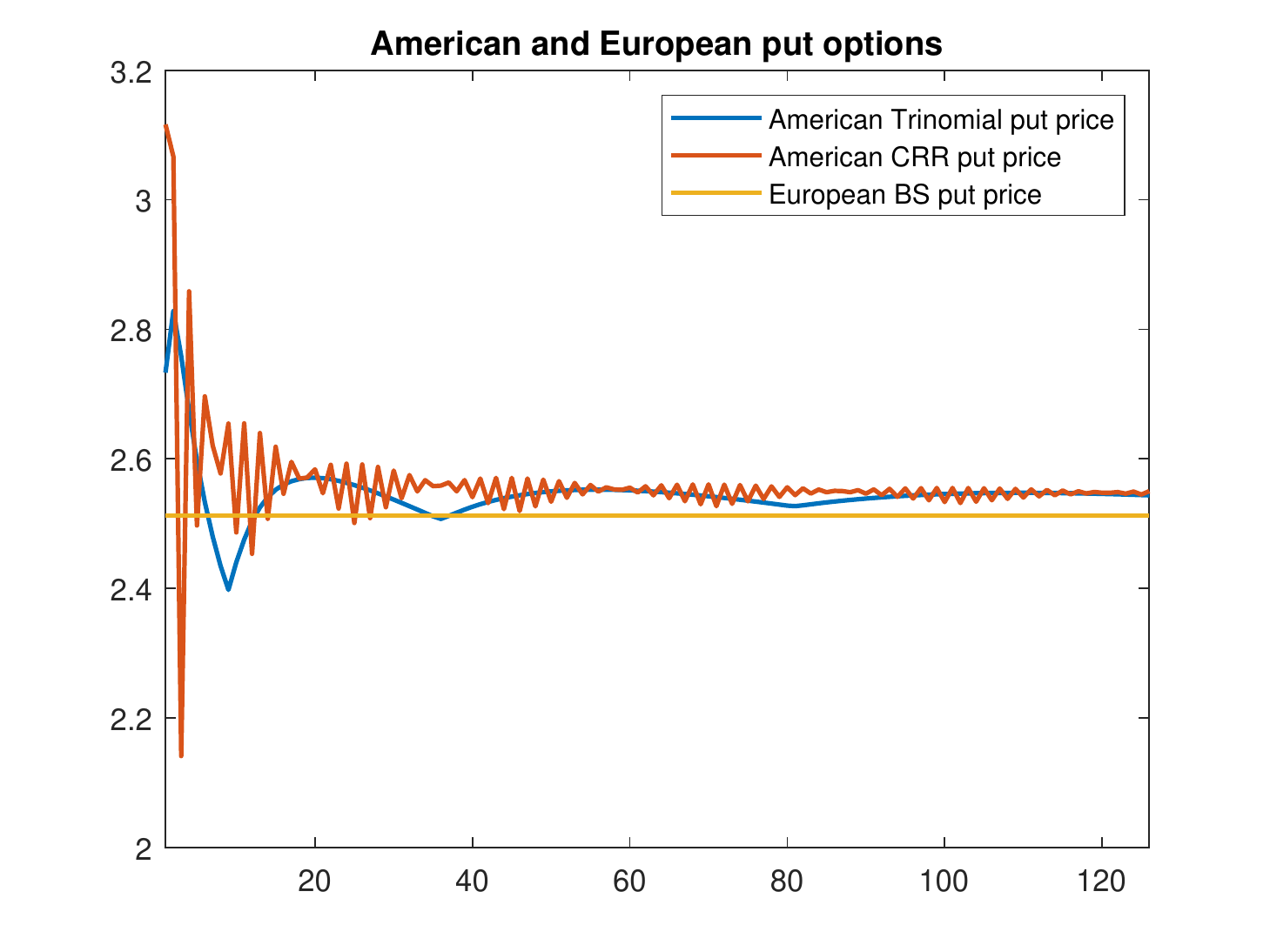}
	\end{subfigure}
	\caption{Trinomial and CRR prices for American call and put options, $K=90$.}
	\label{fig:AmExK90}
\end{figure}
\begin{figure}[t!]
	\begin{subfigure}{.5\linewidth}
		\centering
		\includegraphics[width=\linewidth]{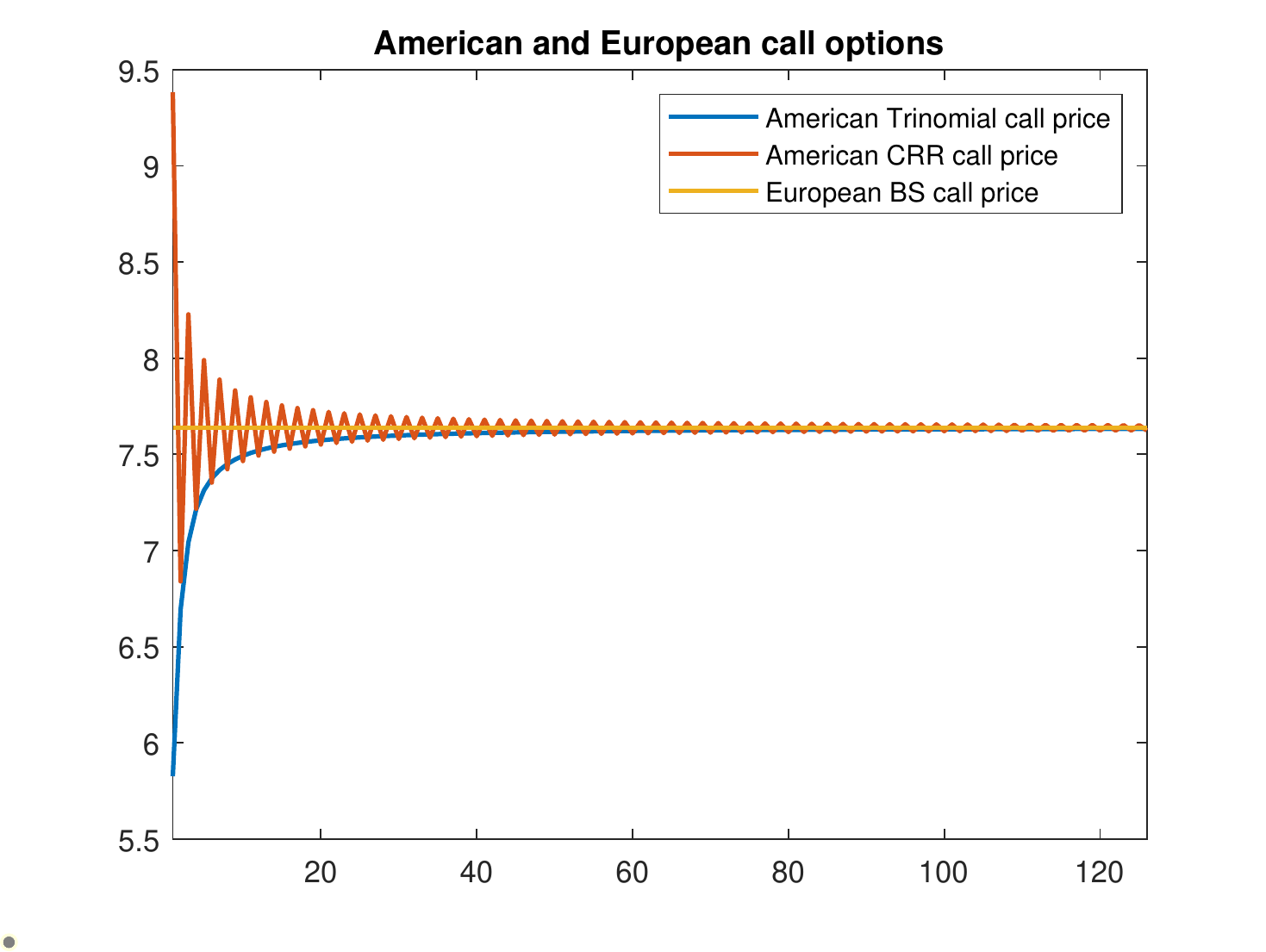}
	\end{subfigure}
	\begin{subfigure}{.5\linewidth}
		\centering
		\includegraphics[width=\linewidth]{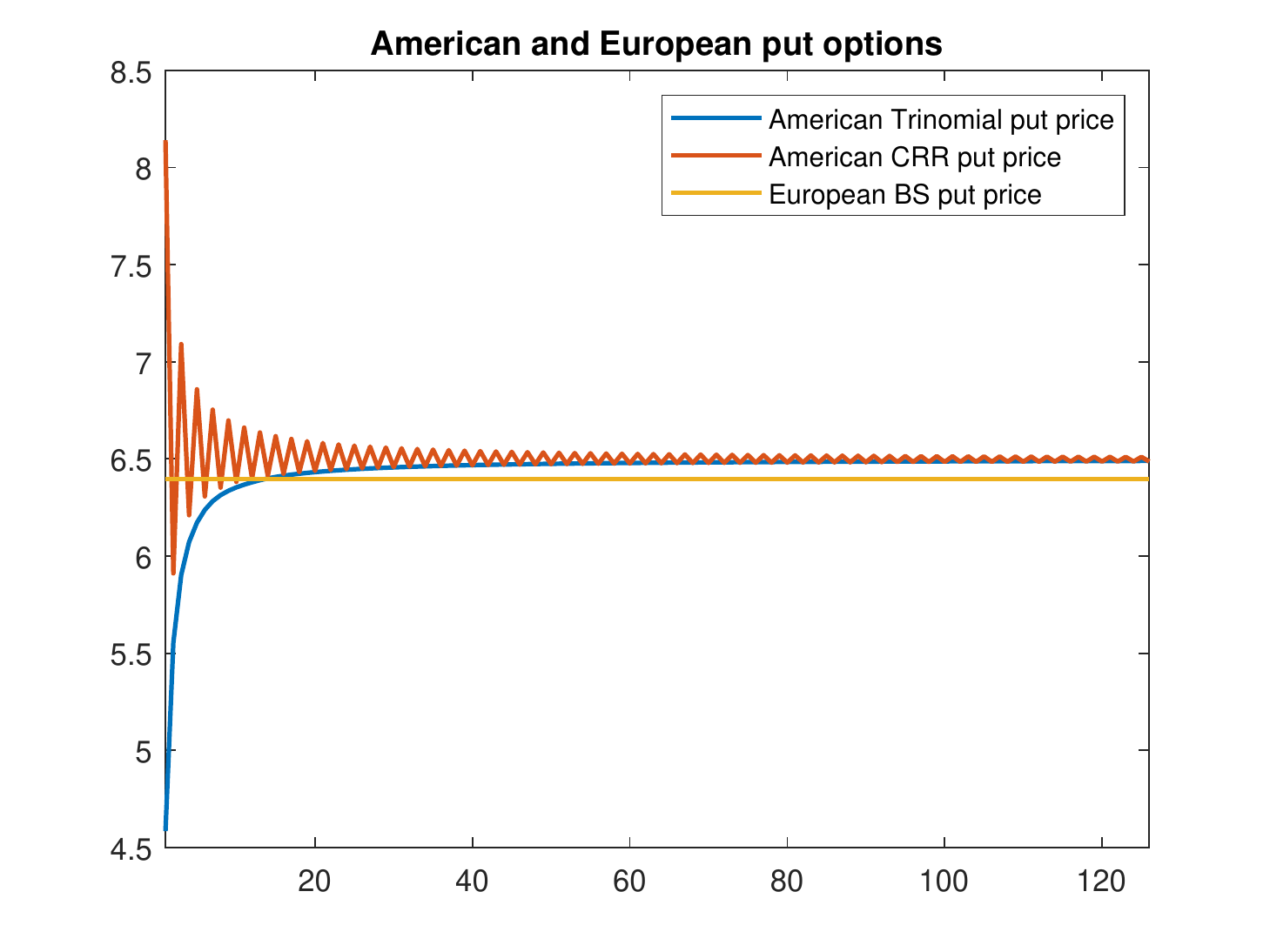}
	\end{subfigure}
	\caption{Trinomial and CRR prices for American call and put options, $K=100$.}
	\label{fig:AmExK100}
\end{figure}
\begin{figure}[t!]
	\begin{subfigure}{.5\linewidth}
		\centering
		\includegraphics[width=\linewidth]{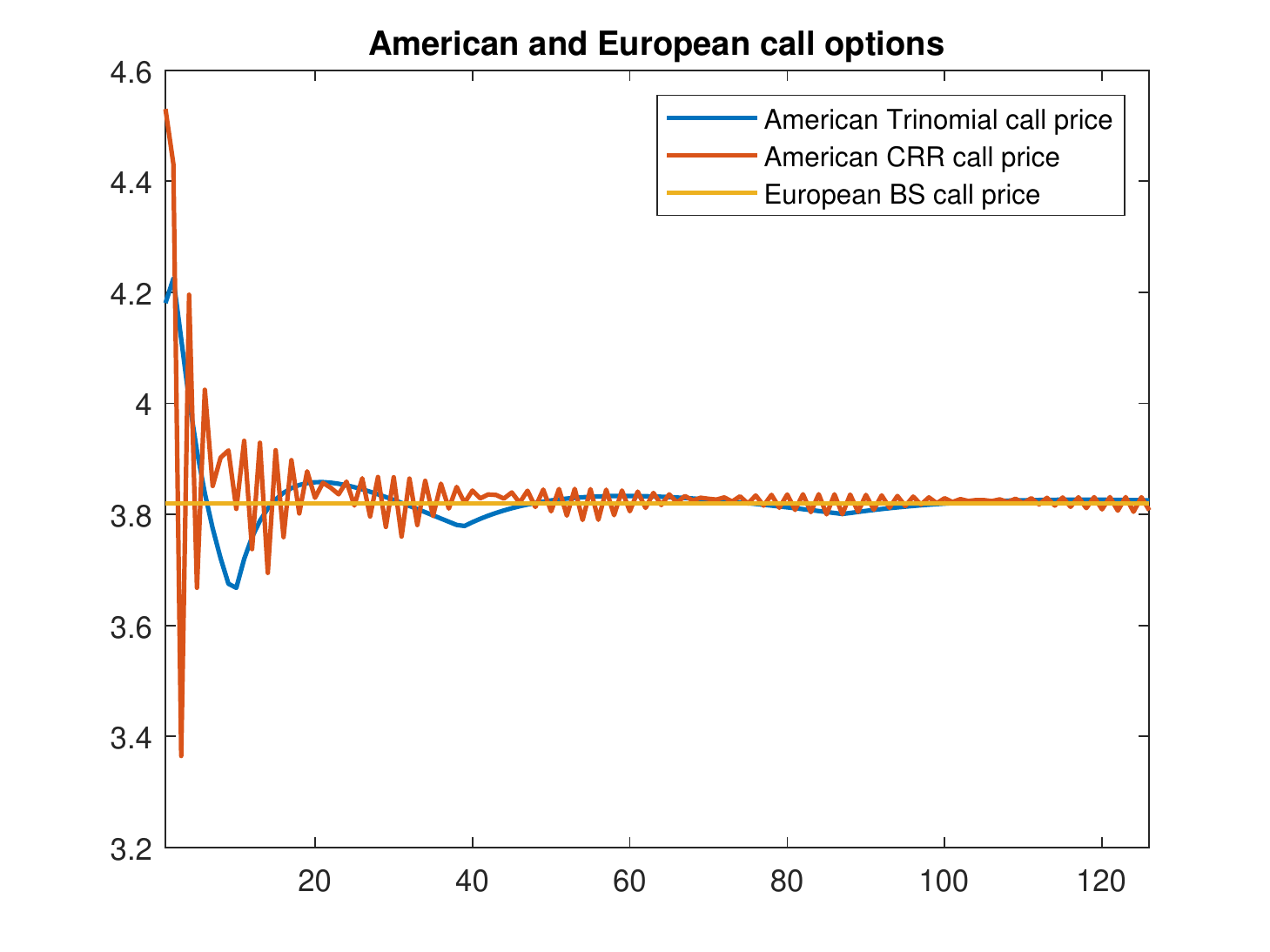}
	\end{subfigure}
	\begin{subfigure}{.5\linewidth}
		\centering
		\includegraphics[width=\linewidth]{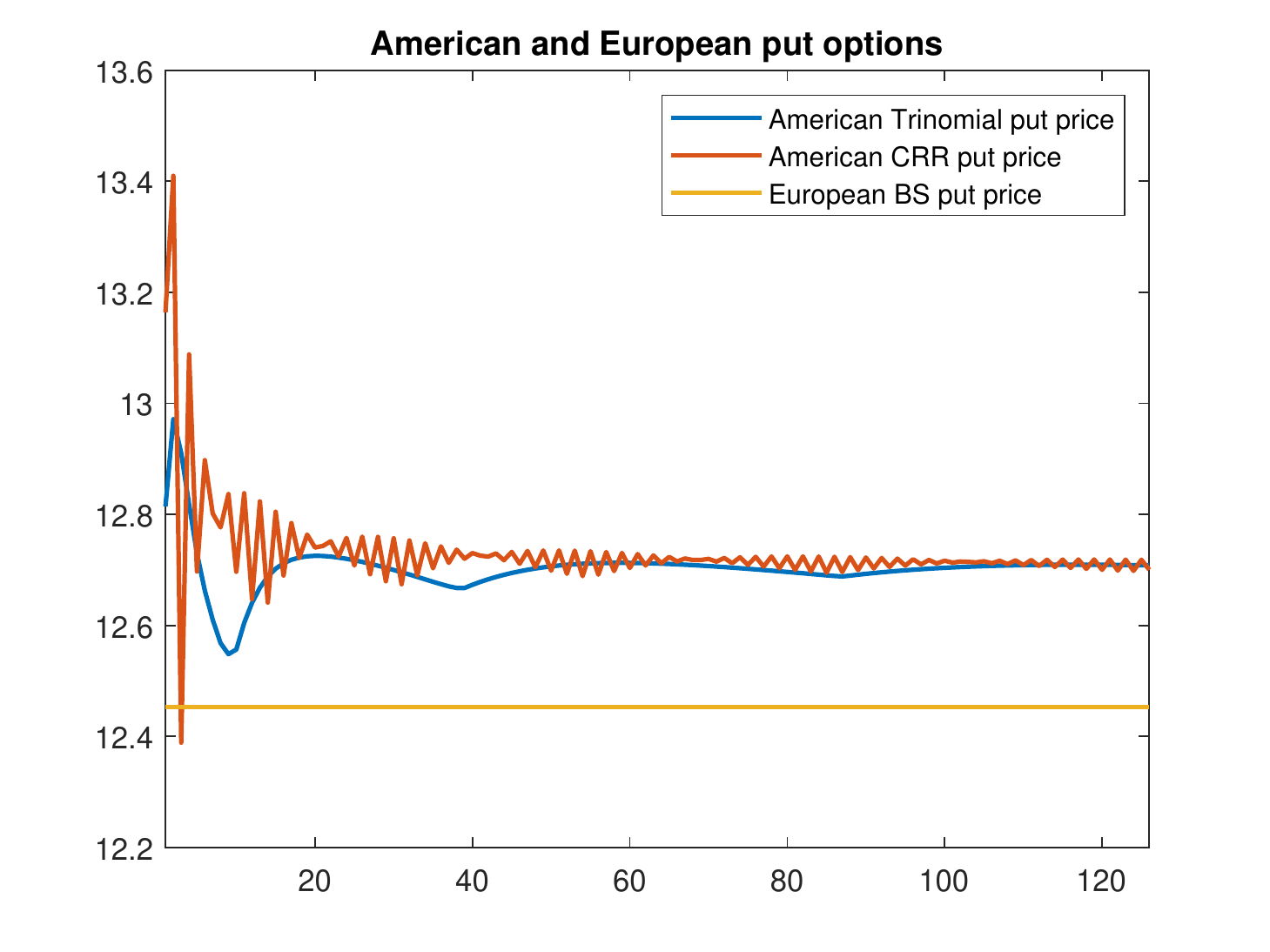}
	\end{subfigure}
	\caption{Trinomial and CRR prices for American call and put options, $K=110$.}
	\label{fig:AmExK110}
\end{figure}

\begin{figure}[t!]
	\begin{subfigure}{.5\linewidth}
		\centering
		\includegraphics[width=\linewidth]{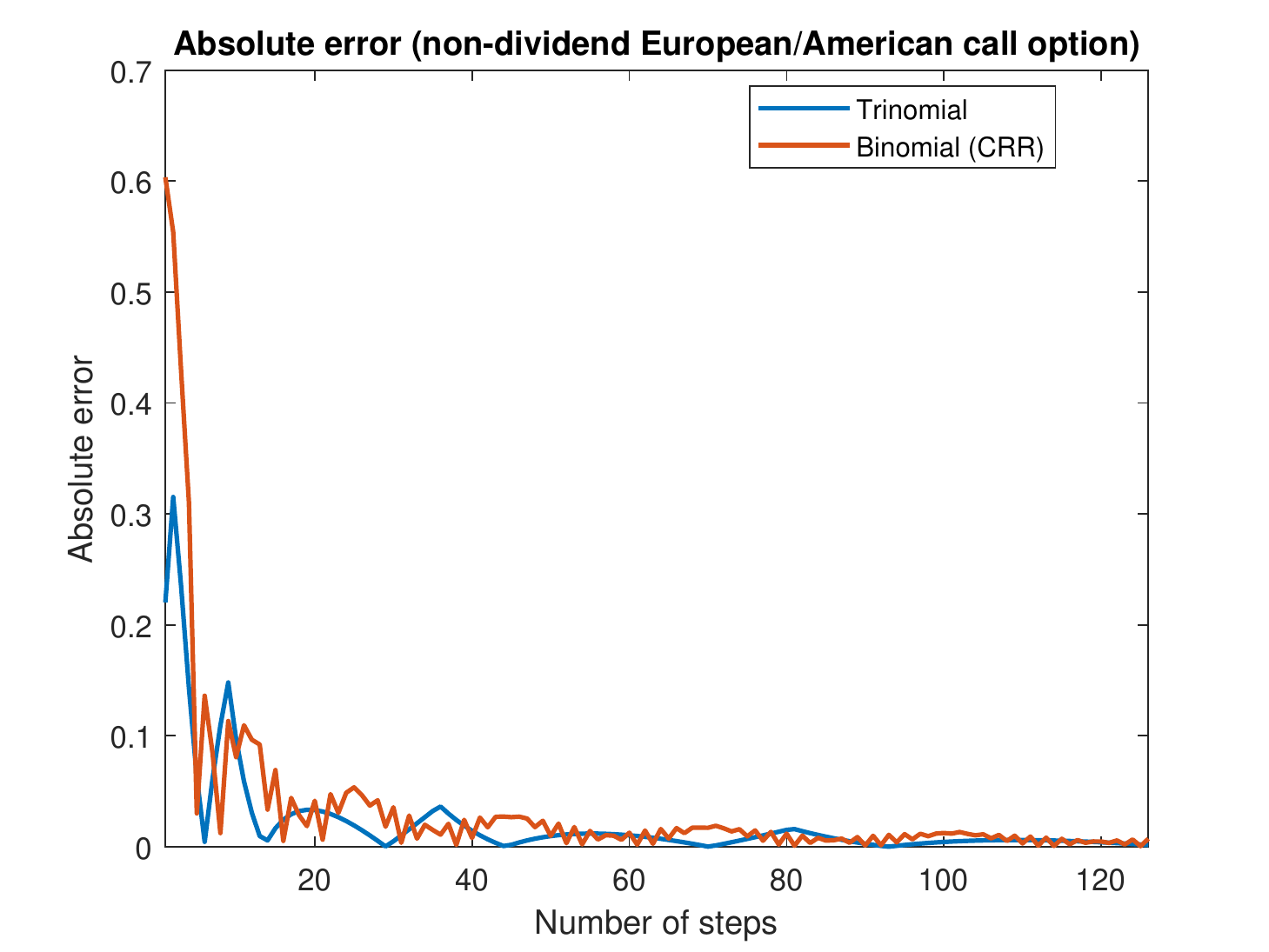}
		\caption{$K=90$.}
	\end{subfigure}
	\begin{subfigure}{.5\linewidth}
		\centering
		\includegraphics[width=\linewidth]{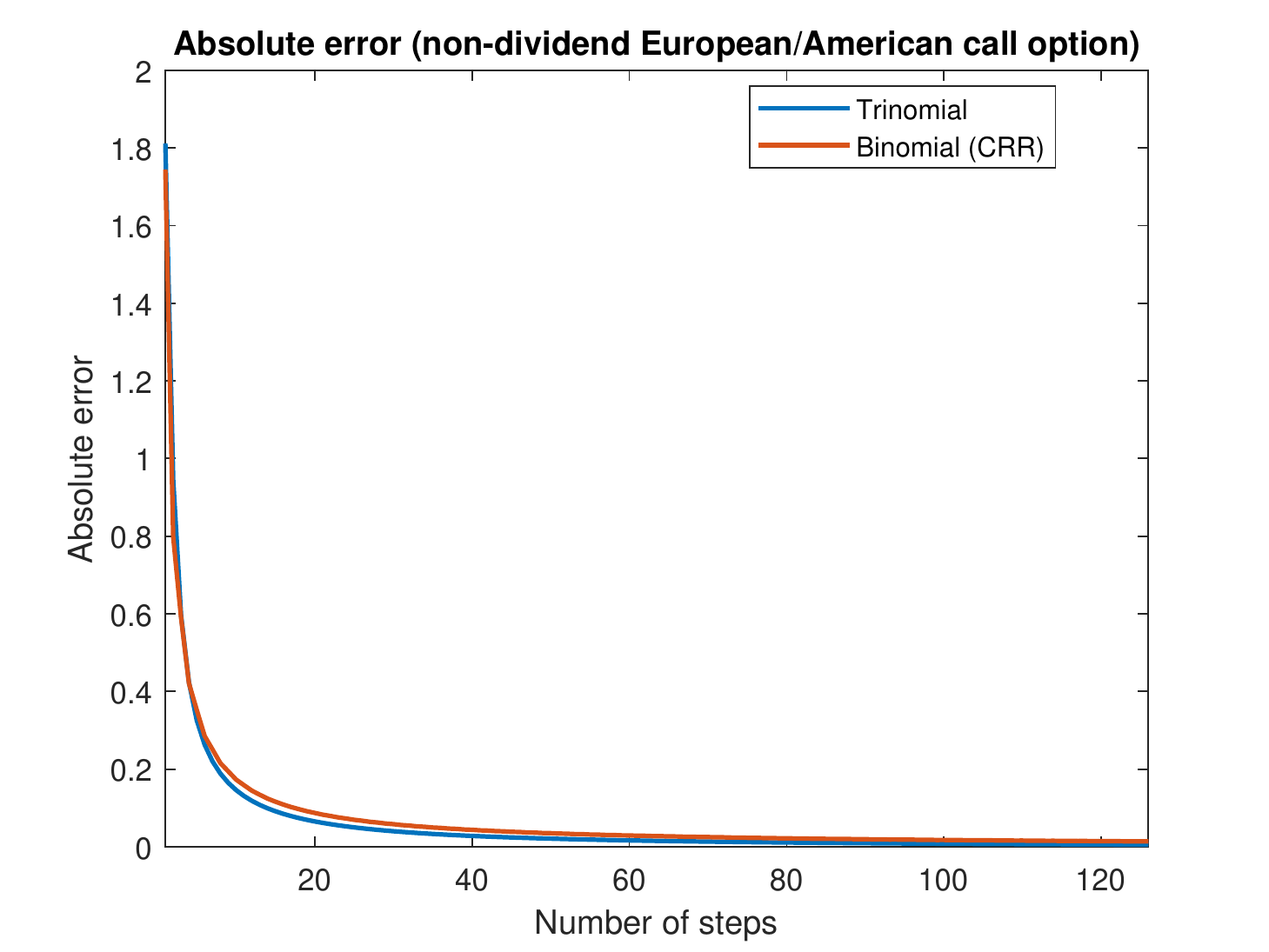}
	\caption{$K=100$.}
	\end{subfigure}
	\caption{Absolute errors of trinomial and CRR prices vs Black--Scholes price, for non-dividend paying European/American call options, $K=\{90,100\}$.}
	\label{fig:abserror}
\end{figure}
\section{Discussion}
\label{Discussion and future works}
In this paper, we briefly reviewed the  cubature method on Wiener space where we specifically applied cubature method and cubature formula on Black--Scholes and Black's models. We saw that using cubature formula of degree 5, solving Black--Scholes and Black's SDEs reduces to solving 3 ordinary differential equations. This approach is not very accurate for long time intervals and therefore we constructed a trinomial tree model for very small time intervals using the result of cubature formula. Then, we find the numerical prices of European call and put options using our developed cubature formula and compare our results with analytical prices of Black--Scholes model. Moreover, we proved that the sequences of constructed trinomial tree converges to the geometric Brownian motion. After that, we studied the  martingale conditions and we extended the results. The extension of results (among other possible applications) included pricing American  options.


\end{document}